\documentclass[a4paper,11pt]{llncs}
\usepackage{authblk}
\usepackage[english]{babel} 
\usepackage[T1]{fontenc}
\usepackage[latin1]{inputenc}
\usepackage{mathtools, amsfonts, amssymb}
\usepackage{mathrsfs}
\usepackage{url}
\usepackage{enumerate}
\usepackage{algorithm, algorithmic}
\usepackage{color}

\spnewtheorem{fact}{Fact}{\bfseries}{\itshape}

\spnewtheorem{notation}{Notation}{\bfseries}{}

\title{On the security of a Loidreau rank metric code based
encryption scheme}

\author{Daniel Coggia\inst{1,2} \and Alain Couvreur\inst{2,3}}
\institute{DGA \and
  INRIA,\and
  LIX, CNRS UMR 7161,
  \'Ecole polytechnique,\break 91128 Palaiseau Cedex, France.\break
  \email{daniel.coggia@inria.fr},
  \email{alain.couvreur@inria.fr}}

\newcommand{\eqdef}{\stackrel{\text{def}}{=}}

\newcommand{\Span}[1]{\mathbf{Span}{\left\{ #1 \right\}}}

\renewcommand{\leq}{\leqslant} 
 
\renewcommand{\geq}{\geqslant}

\newcommand{\Fq}{\ensuremath{\mathbb{F}_q}}

\newcommand{\Fqm}{\ensuremath{\mathbb{F}_{q^m}}}

\newcommand{\F}{\ensuremath{\mathbb{F}}}

\newcommand{\code}[1]{\ensuremath{\mathscr{#1}}}
\newcommand{\Cpub}{{\code{C}_{\textup{pub}}}}
\newcommand{\Csec}{{\code{C}_{\text{sec}}}}

\newcommand{\AC}{\code{A}}

\newcommand{\CC}{\code{C}}

\newcommand{\frob}[2]{#1^{[#2]}}

\newcommand{\av}{\mat{a}}
\newcommand{\bv}{\mat{b}}
\newcommand{\cv}{\mat{c}}

\newcommand{\ev}{\mat{e}}
\newcommand{\gv}{{\mat{g}}}
\newcommand{\hv}{{\mat{h}}}
\newcommand{\mv}{\mat{m}}

\newcommand{\uv}{\mat{u}}

\newcommand{\vv}{\mat{v}}

\newcommand{\xv}{{\mat{x}}}

\newcommand{\yv}{{\mat{y}}}

\newcommand{\rkwt}[1]{\ensuremath{|#1|_{\rm R}}}
\newcommand{\dist}[2]{d_{\rm R}(#1, #2)}
\newcommand{\dmin}[1]{\textrm{d}_{\rm min}(#1)}
\newcommand{\suppt}[1]{\ensuremath{\textsf{supp}(#1)}}
\newcommand{\V}{\mathcal{V}}

\newcommand{\card}[1]{\ensuremath{\left|#1\right|}}
\newcommand{\prob}{\mathbb{P}}

\newcommand{\esp}[1]{\ensuremath{\mathbb{E}\left(#1\right)}}

\newcommand{\GB}[3]{{\left[\begin{array}{c}
                     #1 \\ #2
                     \end{array}\right]}_{#3}}

\newcommand{\mat}[1]{\ensuremath{\boldsymbol{#1}}}

\newcommand{\Gm}{\mat{G}}
\newcommand{\Hm}{\mat{H}}
\newcommand{\Pm}{\mat{P}}
\newcommand{\Qm}{\mat{Q}}

\newcommand{\Gpub}{\Gm_{\rm pub}}
\newcommand{\Hpub}{\Hm_{\rm pub}}
\newcommand{\Gsec}{\Gm_{\rm sec}}
\newcommand{\Hsec}{\Hm_{\rm sec}}

\newcommand{\Mspace}[2]{\mathcal{M}_{#1}(#2)}
\newcommand{\GL}[2]{\mathbf{GL}(#1, #2)}
\newcommand{\PGL}[2]{\mathbf{PGL}(#1, #2)}

\newcommand{\degq}{\textrm{deg}_q}
\newcommand{\qPol}{\mathcal{L}}
\newcommand{\qPold}[1]{\qPol_{< #1}}
\newcommand{\Gab}[2]{\mathcal{G}_{#1}(#2)}

\newcommand{\map}[4]{\left\{
\begin{array}{ccc}
#1 & \longrightarrow & #2 \\
#3 & \longmapsto     & #4
        \end{array}
\right.}

\newcommand{\dimfq}{\textup{dim}_{\Fq}}
\newcommand{\dimfqm}{\textup{dim}_{\Fqm}}
\newcommand{\Cpubd}{{\Cpub^{\bot}}}
\newcommand{\Csecd}{{\Csec^{\bot}}}
\newcommand{\Crand}{{\code{C}_{\textup{rand}}}}

\newcommand{\vspaceof}[1]{\left\langle #1\right\rangle}

\def\longversion{1}

\begin{document}
\maketitle

\begin{abstract}
  We present a polynomial time attack of a rank metric code based
  encryption scheme due to Loidreau for some parameters. 
\end{abstract}


\medskip

 \section*{Introduction}
To instantiate a McEliece encryption scheme, one needs
a family of codes with random looking generator matrices and
an efficient decoding algorithm. While the original proposal
due to McEliece himself \cite{M78} relies on classical Goppa
codes endowed with the Hamming metric, one can actually consider
codes endowed with any other metric. The use of $\F_{q^m}$--linear
rank metric codes,
first suggested by Gabidulin {\em et. al.} \cite{GPT91} is of
particular interest, since the $\F_{q^m}$--linearity permits a
very ``compact'' representation of the code and hence permits to
design a public key encryption scheme with rather short keys
compared to the original McEliece proposal.

Compared to the Hamming metric world, only a few families of codes with
efficient decoding algorithms are known in rank metric. Basically,
the McEliece scheme has been instantiated with two general families of
rank metric codes, namely Gabidulin codes \cite{D78,G85} and
LRPC codes \cite{GMRZ13}.

In \cite{L17}, Loidreau proposed the use of codes which can
in some sense be regarded as an intermediary version between Gabidulin
codes and LRPC codes. These codes are obtained by right
multiplying a Gabidulin code with an invertible matrix whose
entries are in $\F_{q^m}$ and span an $\F_q$--subspace of
small dimension $\lambda$. This approach can be regarded as
a ``rank metric'' counterpart of the BBCRS scheme \cite{BBCRS16}
in Hamming metric.

In the present article, we explain why the case $\lambda = 2$ and
$\dim \Cpub \geq n/2$ is weak and describe a polynomial time key
recovery attack in this situation.

\paragraph{Note.} The material of the present article has been
communicated to Pierre Loidreau in April 2016. The article
\cite{L17} is subsequent to this discussion and proposes parameters
which avoid the attack described in the present article.

 \section{Prerequisites}
\subsection{Rank metric codes}
In this article $m, n$ denote positive integers and $q$ a prime power.
A {\em code of dimension $k$} is an $\Fqm$--subspace of $\Fqm^n$ whose
dimension as an $\Fqm$--vector space is $k$. Given a vector
$\xv \in \Fqm^n$, the {\em rank weight} or {\em rank} of $\xv$,
denoted as $\rkwt{\xv}$ is the dimension of the $\Fq$--vector subspace
of $\Fqm$ spanned by the entries of $\xv$.  The {\em support} of a
vector $\xv \in \Fqm^n$, denoted $\suppt{\xv}$ is the $\Fq$--vector
space spanned by the entries of $\xv$. Hence the rank of $\xv$ is
nothing but the dimension of its support.  The {\em rank distance} or
{\em distance} of two vectors $\xv, \yv \in \Fqm^n$ is defined as
$$
\dist{\xv}{\yv} \eqdef \rkwt{\xv-\yv}.
$$
Given a linear code $\CC \subseteq \Fqm^n$, the {\em minimum distance}
of $\CC$ is defined as
$$
\dmin{\CC} \eqdef \min_{\xv \in \CC \setminus \{0\}}
\left\{\rkwt{\xv}\right\}.
$$
Finally, given a code $\CC \subseteq \Fqm^n$, the {\em dual} code of
$\CC$, denoted by $\CC^\perp$ is the orthogonal of $\CC$ with respect
to the canonical inner product in $\Fqm$:
\[
\map{\Fqm \times \Fqm}{\Fqm}{(\xv, \yv)}{\sum_{i=1}^m x_i y_i.}
\]

\if\longversion1
\begin{remark}
  Note that rank metric codes can be defined in a more general setting
  as subspaces of a space of matrices or spaces of morphisms between
  two vector spaces. A code $\CC \subseteq \Fqm^n$ can be regarded as
  a space of matrices by choosing an $\Fq$--basis of $\Fqm$ and
  associating to each vector $\cv \in \CC$ the matrix whose $i$--th
  column is the decomposition of the entry $c_i$ in the basis.  A
  major difference between general spaces of matrices and the codes we
  introduced is that our vector spaces have an $\Fqm$--linear
  structure. We chose to limit our presentation to codes of the form
  $\CC \subseteq \Fqm^n$ since only these codes are the object of the
  study in the present article.
\end{remark}
\fi

\subsection{$q$--polynomials and Gabidulin codes}
A {\em $q$--polynomial} or a {\em linearized polynomial} is an
$\Fqm$--linear combination of monomials
$X, X^q, X^{q^2}, \ldots, X^{q^{s}}, \ldots$ Such a polynomial induces
a function $\Fqm \rightarrow \Fqm$ which is $\Fq$--linear.  The {\em
  $q$--degree} of a $q$--polynomial $P$, denoted by $\degq (P)$ is the
integer $s$ such that the degree of $P$ is $q^s$. In short:
\[
P = \sum_{i = 0}^{\degq (P)}p_i X^{q^i},\ p_i\in\Fqm,\ p_{\degq
  (P)}\neq 0.
\]
The space of $q$--polynomials is denoted by $\qPol$ and, given a
positive integer $s$, the space of $q$--polynomials of degree less
than $s$ is denoted by
$$
\qPold{s} \eqdef \left\{ P \in \qPol ~|~ \degq (P) < s\right\}.
$$
Given positive integers $k, n$ with $k \leq n \leq m$ and an
$n$--tuple $\av = (a_1, \ldots, a_n)$ of $\Fq$--linearly independent
elements of $\Fqm$, the {\em Gabidulin code} $\Gab{k}{\av}$ is defined
as
$$
\Gab{k}{\av} \eqdef \{ (f(a_1), \ldots, f(a_n)) ~|~ f \in \qPold{k}
\}.
$$
\if\longversion1 This code has a generator matrix of the form:
$$
\begin{pmatrix}
  a_1 & \ldots & a_n\\
  a_1^q & \ldots & a_n^q \\
  \vdots & & \vdots \\
  a_1^{q^{k-1}} & \ldots & a_n^{q^{k-1}}
\end{pmatrix}.
$$
\fi Such codes are known to have minimum distance $n-k+1$ and to
benefit from a decoding algorithm correcting up to half the minimum
distance (see \cite{L06a}).

Note that the vector $\av$ is not unique as shown by the following
lemma which will be useful for our attack.

\begin{lemma}[{\cite[Theorem 2]{B03}}]\label{lem:supports}
  Let $\alpha \in \Fqm \setminus \{0\}$. Then
  $ \Gab{k}{\av} = \Gab{k}{\alpha\av}.  $
\end{lemma}

\subsection{The component-wise Frobenius map}
In what follows, we will frequently apply the component-wise Frobenius
map or its iterates to vectors or codes. Hence, we introduce the
following notation. Given a vector $\vv \in \Fqm^n$ and a non-negative
integer $s$, we denote by $\frob{\vv}{s}$ the vector:
$$
\frob{\vv}{s} \eqdef (v_1^{q^s}, \ldots, v_n^{q^s}).
$$
Similarly, given a code $\CC \subseteq \Fqm^n$ and a positive
integer $s$, the code $\frob{\CC}{s}$ denotes the code
$$
\frob{\CC}{s} \eqdef \{\frob{\cv}{s} ~|~ \cv \in \CC\}.
$$

\subsection{Overbeck's distinguisher}
In \cite{O08}, Overbeck proposes a general framework to break cryptosystems
based on Gabidulin codes.
The core of his attack is that a simple operation allows us to distinguish
Gabidulin codes from random ones. Indeed, given a random
code $\CC \subseteq \Fqm^n$ of dimension $k < n/2$, the expected
dimension of the code $\CC + \frob{\CC}{1}$ equals $2k$ and, equivalently
$\CC \cap \frob{\CC}{1}$ is likely to be equal to $0$. More generally,
we have the following statement.

\begin{proposition}\label{prop:random}
  If $\Crand$ is a code of length $n$ and dimension $k$ chosen
  uniformly at random, then for a non-negative integer $a$ and for a
  positive integer $s < k$, we have
\[
\prob \bigg(\dimfqm\ \Crand + \frob{\Crand}{1}+\cdots+ \frob{\Crand}{s} \leq \min(n, (s+1)k) - a\bigg) = O(q^{-ma}).
\]
\end{proposition}

\if\longversion1
\begin{proof}
  See Appendix~\ref{subsec:proof_sum_frob}.\qed
\end{proof}
\fi

On the other hand, for a Gabidulin code, the behaviour with respect to
such operations is completely different, as explained in the following
statement.

\begin{proposition}
  \label{prop:gab_frob}
  Let $\av\in\Fqm^n$ be a word of rank $n$, $k\leq n$ and $s$ be an
  integer. Then,
  \begin{eqnarray}
    \label{eq:classic} \Gab{k}{\av} \cap \frob{\Gab{k}{\av}}{1} &=&
                                         \Gab{k-1}{\frob{\av}{1}}; \\
    \label{eq:overbeck} \Gab{k}{\av} + \cdots + \frob{\Gab{k}{\av}}{s} &=&
                                         \Gab{k+s}{\av}.
  \end{eqnarray}
\end{proposition}

\begin{proof}
  For (\ref{eq:overbeck}) see \cite[Lemma 5.1]{O08}. Identity
  (\ref{eq:classic})
   that can be proved as
  follows: for inclusion ``$\varsupsetneq$'', one checks that
  the generators $\frob{\av}{i}$ for $i = 1, \dots ,k-1$ of
  $\Gab{k-1}{\frob{\av}{1}}$ are in $\Gab{k}{\av} \cap \frob{\Gab{k}{\av}}{1}$.  Conversely, the two codes have the same
  dimension because of (\ref{eq:overbeck}) for $s = 1$. \qed
\end{proof}

\if\longversion1
\begin{example}
  As an illustration of the difference of behaviour of Gabidulin codes
  compared to random codes, given a Gabidulin code $\code{G}$ of length $n$
  and dimension $k < n/2$, then $\dim \code{G}+ \frob{\code{G}}{1} = k+1$,
  while for a random code $\Crand$ of the same length and dimension,
  $\prob\big(\dim \Crand + \frob{\Crand}{1} = 2k\big)$ tends to
  $1$ when $m$ tends to infinity.
\end{example}
\fi

 \section{Loidreau's scheme}
In order to mask the structure of Gabidulin codes and to resist to
Overbeck's attack, Loidreau suggested in \cite{L17} the following
construction.  Denote by $\Gm$ a random generator matrix of a Gabidulin code
$\Gab{k}{\bv}$.  Fix an integer $\lambda \leq m$  and an $\Fq$--vector
subspace $\V$ of $\Fqm$ of dimension $\lambda$.
Let $\Pm \in \GL{n}{\Fqm}$ whose entries are all in $\V$. Then, let
$$
\Gpub \eqdef \Gm \Pm^{-1}.
$$
We have the following encryption scheme.
\begin{description}
\item[{\bf Public key:}] The pair $(\Gpub, t)$ where 
  $t \eqdef \lfloor \frac{n-k}{2\lambda} \rfloor$.
\item[{\bf Secret key:}] The pair $(\bv, \Pm)$.
\item[{\bf Encryption:}] Given a plain text $\mv \in \Fqm^k$, choose a
  uniformly random vector $\ev \in \Fqm^n$ of rank weight $t$. The cipher text
  is
  $$
  \cv \eqdef  \mv \Gpub + \ev.
  $$
\item[{\bf Decryption:}] Compute,
  $$
  \cv \Pm = \mv \Gm + \ev \Pm.
  $$
  Since the entries of $\Pm$ are all in $\V$ then, the entries
  of $\ev \Pm$ are in the product space $\suppt{\ev} \cdot \V
  \eqdef \langle uv ~|~ u\in \suppt{\ev},\
  v\in \V\rangle_{\Fq}$. The
  dimension of this space is bounded from above by
  $t\lambda \leq \frac{n-k}{2}$. Therefore, using a classical decoding
  algorithm for Gabidulin codes, one can recover $\mv$.
\end{description}

 \section{A distinguisher}
\label{sec:distinguisher}
\subsection{Context}
The goal of this section is to establish a distinguisher for
Loidreau's cryptosystem instantiated with $\lambda = 2$ and a public
code $\Cpub$ of dimension $k \geq \frac{n}{2}$. Similar to Overbeck's
attack, this distinguisher relies on Propositions~\ref{prop:random}
and \ref{prop:gab_frob}.

Similar to the attacks of BBCRS system~\cite{CGGOT14,COTG15}, it is
more convenient to work on the dual of the public code because of the
following lemmas.  \if\longversion1 By {\em dual}, we mean the
orthogonal code with respect to the canonical inner product
\[
\map{\Fqm^n \times \Fqm^n}{\Fqm}{(\xv, \yv)}{
  \sum_{i=1}^n x_i y_i}.
\]
\fi
following lemmas.

\begin{lemma}[{\cite[Theorem~7]{G85}}]
  \label{lem:gab_dual}
  Let $\bv\in\Fqm^n$ of rank $n$. The code $\Gab{k}{\bv}^\perp$ is a
  Gabidulin code $\Gab{n-k}{\av}$ for some $\av\in\Fqm^n$ of rank $n$.
\end{lemma}

\begin{lemma}
  \label{lem:hpub_decomposition}
  Any full--rank generator matrix $\Hpub$ of $\Cpubd$ can be
  decomposed as \if\longversion1
  \[ \else $ \fi \Hpub = \Hsec \Pm^T \if\longversion1 \]
  \else $
  \fi where $\Hsec$ is a parity--check matrix of the Gabidulin code
  $\Gab{k}{\bv}$.
\end{lemma}

\begin{proof}
  Let $\Hpub$ be a full--rank generator matrix of $\Cpubd$. Since
  $\Cpub=\Csec P^{-1}$, there exists a full--rank generator matrix
  $\Gsec$ of $\Csec$ such that
  \[\Hpub(\Gsec P^{-1})^T = 0 \qquad\text{then}\qquad\Hpub (P^{-1})^T
  \Gsec^T= 0.\]
  Which means that $\Hpub (P^{-1})^T$ is a parity-check of
  $\Csec = \Gab{k}{\bv}$.\qed
\end{proof}

The convenient aspect of the previous lemma is that the matrix $\Pm$
has its entries in a small vector space, while its inverse has not.

\subsection{The case $\lambda = 2$}\label{ss:dist_l_2}

We suppose in this section that the vector space $\V \subseteq \Fqm^n$ in which
the matrix $\Pm$ has all its entries has dimension $2$:
\[\lambda = \dimfq \V = 2.\]
Note that, w.l.o.g, one can suppose
that $1 \in \V$. Indeed, if $\V$ is spanned over $\Fq$ by $\alpha, \beta
\in \Fqm\setminus \{0\}$, then one can replace $\Hsec$ by $\Hsec' =
\alpha \Hsec$
which spans the {\bf same code}, $\Pm' = \alpha^{-1} \Pm$ has entries in
$\V' = \Span{1, \alpha^{-1}\beta},$ and $\Hpub = \Hsec' \Pm'^T$.

Thus, from now on, we suppose that $\V = \Span{1, \gamma}$ for some $\gamma \in \Fqm \setminus \Fq$.
Consequently, $\Pm^T$ can be decomposed as
\[
  \Pm^T = \Pm_0 + \gamma \Pm_1,
\]
where $\Pm_0, \Pm_1$ are square matrices with entries in $\Fq$.

We have seen that $\Csecd = \Gab{n-k}{\av}$ for some $\av\in\Fqm^n$
with $\rkwt\av = n$.  We define
\[
  \gv \eqdef \av \Pm_0 \quad {\rm and} \quad   \hv \eqdef \av \Pm_1.
\]
\begin{lemma}\label{lem:basis_Cpubd}
The code $\Cpubd$ is spanned by
\[
  \gv + \gamma \hv, \frob{\gv}{1} + \gamma \frob{\hv}{1},
  \ldots , \frob{\gv}{n-k-1} + \gamma \frob{\hv}{n-k-1}.
\]  
\end{lemma}

\if\longversion1
\begin{proof}
For any $\cv\in\Cpubd$
there exists $P\in\qPold{n-k}$ such that
\[
\cv = P(\av)\Pm^T = P(\av)\Pm_0 + \gamma P(\av)\Pm_1 = P(\gv) + \gamma P(\hv),
\]
which yields the result.  \qed
\end{proof}
\fi

We can now state a crucial result.
\begin{theorem}\label{thm:main}
  The dual of the public code satisfies:
  \[
    \dimfqm\ \Cpubd + \frob{\Cpubd}{1} + \frob{\Cpubd}{2} \leq 2\ \dimfqm\ \Cpubd + 2.
  \]
\end{theorem}

\if\longversion1
\begin{proof}
  Thanks to Lemma~\ref{lem:basis_Cpubd},
  we prove that $\Cpubd  + \frob{\Cpubd }{1}$ is spanned by
\[
  \gv + \gamma \hv, \cdots, \frob{\gv}{n-k-1} + \gamma \frob{\hv}{n-k-1}
  \quad {\rm and} \quad
  \frob{\gv}{1} + \gamma^q \frob{\hv}{1}, \ldots, \frob{\gv}{n-k} +
  \gamma^q \frob{\hv}{n-k}.
\]
Equivalently, $\Cpubd + \frob{\Cpubd}{1}$ is spanned by:
\[
  \gv + \gamma \hv \quad {\rm and} \quad
  \frob{\gv}{1}, \frob{\hv}{1}, \ldots,
  \frob{\gv}{n-k-1}, \frob{\hv}{n-k-1}
  \quad {\rm and} \quad
  \frob{\gv}{n-k} + \gamma^q \frob{\hv}{n-k}.
\]
Finally, a similar reasoning allows us to show that
$\Cpubd + \frob{\Cpubd}{1} + \frob{\Cpubd}{2}$ is spanned by
\[
  \gv + \gamma \hv \quad {\rm and} \quad
  \frob{\gv}{1}, \frob{\hv}{1}, \ldots,
  \frob{\gv}{n-k}, \frob{\hv}{n-k}
  \quad {\rm and} \quad
  \frob{\gv}{n-k+1} + \gamma^{q^2} \frob{\hv}{n-k+1},
\]
and hence has dimension at most $2(n-k)+2$.
  \qed
\end{proof}
\fi

As a conclusion, thanks to Proposition~\ref{prop:random},
we deduce that $\Cpubd$ is distinguishable in polynomial time
from a random code as soon as \mbox{$2k-2>n$} and $k < n-2$.

\subsection{The case $\lambda > 2$}
In the general case, the approach consists in computing
\[
  \Cpub^\perp + \frob{{\Cpub^\perp}}{1} + \cdots +
  \frob{{\Cpub^\perp}}{\lambda}.
\]
A similar reasoning permits to prove that the dimension of this space
is bounded from above by $\lambda \dim \Cpub^\perp + \lambda$.
Thanks to Proposition~\ref{prop:random}, such a code is distinguishable
from a random one if
\[
  \lambda\dim \Cpub^\perp + \lambda < \min \{n, (\lambda + 1)\dim \Cpub^\perp\}
\]
That is to say
\[
  \lambda (n-k+1) < n \qquad {\rm and} \qquad \lambda (n-k+1) <
  (\lambda + 1)(n-k).
\]
In summary, the public code is distinguishable from a random one if its dimension $k$ satisfies:
\[
  n\left(1 - \frac 1 \lambda\right) + 1 < k < n- \lambda.
\]

 \section{The attack}
In this section, we derive an attack from the distinguisher defined in
Section~\ref{sec:distinguisher}.
In what follows, we suppose that $\lambda = 2$ and the public code
has rate larger than $1/2$ so that the distinguisher introduced in
Section~\ref{sec:distinguisher} works on it.
Recall that $\Cpubd = \Gab{n-k}{\av}\Pm$ for some $\av \in \F_{q^m}^n$
whose entries are $\Fq$--independent and $\Pm$ is of the form
$\Pm_0 + \gamma \Pm_1$ for $\Pm_0, \Pm_1 \in \Mspace{n}{\F_q}$
and $\gamma \in \Fqm \setminus \Fq$. Finally recall that
\[
  \gv \eqdef \av \Pm_0 \qquad {\rm and} \qquad \hv \eqdef \av \Pm_1.
\]
In addition, we make the following assumptions:
\begin{enumerate}[(1)]
\item\label{item:assump1} 
  $\displaystyle \Gab{n-k+2}{\gv}\cap\Gab{n-k+2}{\hv}=\{0\}$
  and $\rkwt{\gv}\geq n-k+2$ and $\rkwt{\hv}\geq n-k+2$;
\item\label{item:assump2} $m > 2$ and $\gamma$ is not contained in any
  subfield of $\Fqm$.
\end{enumerate}

\noindent {\bf Comments on the assumptions.}
According to our experiments using {\sc Magma}~\cite{BCP97},
Assumption~(\ref{item:assump1}) holds with a high probability.  For
instance, we ran $1000$ tests for parameters : $q = 2, m = n = 30$ and
$k = 17$. The average rank of $\gv, \hv$ is $29.1$ and their minimum
rank is $27$. In addition for these $1000$ tests, the intersection
$\Gab{n-k+2}{\gv}\cap\Gab{n-k+2}{\hv}$ was always $\{0\}$.
Finally, Assumption~(\ref{item:assump2}) is reasonable in order to
prevent against possible attacks based on an exhaustive search of
$\gamma$.

\bigskip

The aim of the attack is to recover the triple $(\gamma, \gv, \hv)$,
or more precisely, to recover a triple $(\gamma', \gv', \hv')$ such
that
\begin{equation}\label{eq:triple}
  \Cpubd = \vspaceof{\gv'^{[i]} + \gamma' \hv'^{[i]} ~\Big|~
    i = 0, \dots, n-k-1}.
\end{equation}
Actually, the triple $(\gamma, \gv, \hv)$ is far from being unique and
any other triple satisfying~(\ref{eq:triple}) allows us to decrypt
messages (see further \S~\ref{ss:end}).
Let us describe an action of
$\PGL{2}{\F_q}$ on such triples.

\begin{proposition}\label{prop:triples}
  Let $a, b, c, d \in \Fq$ such that $ad-bc \neq 0$ and
  $\delta \in \Fqm$ such that
  $\gamma = \frac{a\delta + b}{c\delta +d}$. Then, the triple
  $(\delta, d\gv + b \hv, c \gv + a \hv)$ satisfies~(\ref{eq:triple}).
\end{proposition}

\if\longversion1
\begin{proof}
  It suffices to observe that for any $i \geq 0$,
  \[
    \gv^{[i]} + \frac{a \delta + b}{c \delta + d} \hv^{[i]} =
      \frac{1}{c \delta + d}\left((d \gv + b \hv)^{[i]}+
      \delta (c\gv + a \hv)^{[i]}\right).
  \]
  \qed
\end{proof}
\fi

\subsection{Step 1: using the distinguisher to compute some
  subcodes}
As shown in the proof of Theorem~\ref{thm:main},
$\Cpubd + \frob{\Cpubd}{1}$ is spanned by:
\[
  \gv + \gamma \hv \quad {\rm and} \quad
  \frob{\gv}{1}, \frob{\hv}{1},
  \ldots,
  \frob{\gv}{r}, \frob{\hv}{r}
  \quad {\rm and} \quad
  \frob{\gv}{r+1} + \gamma^q \frob{\hv}{r+1},
\]
where $r \eqdef n-k-1$.

\begin{lemma}
Under Assumption~(\ref{item:assump1}), we have that
\[(\Cpubd + \frob{\Cpubd}{1}) \cap (\frob{\Cpubd}{1} + \frob{\Cpubd}{2})\]
is spanned by:
\[
  \frob{\gv}{1} + \gamma^q \frob{\hv}{1}
  \quad {\rm and} \quad
  \frob{\gv}{2}, \frob{\hv}{2},
  \ldots,
  \frob{\gv}{r}, \frob{\hv}{r}
  \quad {\rm and} \quad
  \frob{\gv}{r+1} + \gamma^q \frob{\hv}{r+1},
\]  
\end{lemma}

\begin{proof}
Indeed, if a vector $\xv$ belongs to this intersection, we have that
\begin{align*}
  \begin{split}
    \xv ={}& x_{0,0}(\gv + \gamma \hv) + \sum_{i=1}^r
    \left(x_{0,i}\frob{\gv}{i} + y_{0,i} \frob{\hv}{i}\right)
    \\ &+ x_{0,r+1}\left(\frob{\gv}{r+1}+ \gamma^q
      \frob{\hv}{r+1}\right)
  \end{split}\\
  \begin{split}
    ={}& x_{1,0}\left(\frob{\gv}{1} + \gamma^q
      \frob{\hv}{1}\right) + \sum_{i=1}^r
    \left(x_{1,i}\frob{\gv}{i+1} + y_{1,i}
      \frob{\hv}{i+1}\right) \\ &+ x_{1,r+1}\left(\frob{\gv}{r+2}+
      \gamma^{q^2} \frob{\hv}{r+2}\right).
  \end{split}
\end{align*}
Hence
\begin{align*}
\begin{split}
{}&x_{0,0} \gv
+ (x_{0,1}-x_{1,0}) \frob{\gv}{1}
+ \sum_{i=2}^r (x_{0,i}-x_{1,i-1}) \frob{\gv}{i}\\
&+ (x_{0,r+1}-x_{1,r}) \frob{\gv}{r+1}
- x_{1,r+1}\frob{\gv}{r+2}
\end{split}\\
\begin{split}
={}&- x_{0,0} \gamma \hv
+ (\gamma^qx_{1,0}-y_{0,1}) \frob{\hv}{1}
+ \sum_{i=2}^r (y_{1,i-1}-y_{0,i}) \frob{\hv}{i}\\
&+ (y_{1,r}-x_{0,r+1}\gamma^q) \frob{\hv}{r+1}
+ x_{1,r+1} \gamma^{q^2}\frob{\hv}{r+2}.
\end{split}
\end{align*}

Since by Assumption~(\ref{item:assump1})
$\Gab{r+3}{\gv}\cap\Gab{r+3}{\hv}=\{0\}$, 
both sides of this equation are zero and in
particular, $x_{0,0}= x_{1, r+1} = 0$ and
$x_{0,1}=x_{1,0}=\gamma^{-q}y_{0,1}$.  This finally gives that
\[
  \xv\in \vspaceof{
  \frob{\gv}{1} + \gamma^q \frob{\hv}{1},
  \frob{\gv}{2}, \frob{\hv}{2},
  \ldots,
  \frob{\gv}{r}, \frob{\hv}{r},
  \frob{\gv}{r+1} + \gamma^q \frob{\hv}{r+1}}.
\]
\qed
\end{proof}

Then, a proof by induction shows that iterating intersections
\[(\Cpubd + \frob{\Cpubd}{1}) \cap (\frob{\Cpubd}{1} + \frob{\Cpubd}{2}) \cap \cdots
  \cap (\frob{\Cpubd}{r} + \frob{\Cpubd}{r+1}),
\]
yields the code spanned by
\[
\frob{\gv}{r} + \gamma^{q^r} \frob{\hv}{r} \quad {\rm and} \quad
\frob{\gv}{r+1} + \gamma^{q} \frob{\hv}{r+1}.
\]
Applying the inverse of the $r$--th Frobenius, we get the code spanned
by
\[
\gv + \gamma \hv \quad {\rm and} \quad
\frob{\gv}{1} + \gamma^{q^{1-r}} \frob{\hv}{1}.
\]
Next, one can compute
\begin{equation}\label{eq:code_1}
\Cpubd \cap \vspaceof{ \gv + \gamma \hv, \frob{\gv}{1} +
\gamma^{q^{1-r}}\frob{\hv}{1}} = \vspaceof{\gv + \gamma \hv},
\end{equation}
Indeed, if $\xv\in\Cpubd \cap \vspaceof{ \gv + \gamma \hv, \frob{\gv}{1} +
\gamma^{q^{1-r}}\frob{\hv}{1}}$, we can write for some $x$ and $y$,
\[
\xv = \sum_{i=0}^r
x_i\left(\frob{\gv}{i}+\gamma\frob{\hv}{i}\right)= y_0(\gv+\gamma\hv)
+ y_1\left(\frob{\gv}{1}+\gamma^{q^{1-r}}\frob{\hv}{1}\right).
\]
Assumption~(\ref{item:assump1}) gives that $x_1=y_1$ and
$x_1\gamma=y_1\gamma^{q^{1-r}}$. Finally with
Assumption~(\ref{item:assump2}), we deduce $y_1=0$ and hence
$\xv\in\vspaceof{\gv+\gamma\hv}$.

Now, we can compute
\begin{eqnarray}
  \nonumber
  \vspaceof{ \gv + \gamma \hv, \frob{\gv}{1} +
  \gamma^{q^{1-r}}\frob{\hv}{1} } &+& \frob{\vspaceof{ \gv + \gamma \hv
                                      }}{1} \\
  \nonumber
                                  & = & \vspaceof{ \gv + \gamma \hv,
                                        \frob{\gv}{1} + \gamma^{q^{1-r}}\frob{\hv}{1},
                                        \frob{\gv}{1} + \gamma^{q}\frob{\hv}{1}
                                        }\\
  \nonumber   & = & \vspaceof{  \gv + \gamma \hv, \frob{\gv}{1}, \frob{\hv}{1}
                    }.
\end{eqnarray}
Similarly, we compute the intersection with
$\frob{\Cpubd}{-1}\eqdef \frob{\Cpubd}{m-1}$ and get with Assumptions~(\ref{item:assump1}) and (\ref{item:assump2})
\begin{equation}\label{eq:code_2}
\frob{\Cpubd}{-1} \cap \vspaceof{ \gv + \gamma \hv, \frob{\gv}{1} ,
\frob{\hv}{1} } = \vspaceof{ \frob{\gv}{1} + \gamma^{q^{m-1}}
\frob{\hv}{1} }.
\end{equation}
Applying the inverse Frobenius to the last code, we get
$\vspaceof{ \gv + \gamma^{q^{m-2}}\hv}$. Since,
from~(\ref{eq:code_1}), we also know
$\vspaceof{ \gv + \gamma \hv }$, one can compute
\begin{equation}\label{eq:code_2}
\vspaceof{ \gv + \gamma \hv
} +
{\vspaceof{ \frob{\gv}{1} + \gamma^{q^{m-1}}
\frob{\hv}{1} }}^{[-1]}  = 
\vspaceof{ \gv + \gamma \hv, \gv + \gamma^{q^{m-2}} \hv}
 = \vspaceof{ \gv, \hv }.
\end{equation}

\noindent Next, 
for any $i \in \{0, \ldots, r\}$, one can compute
\[
  \Cpub^\perp \cap \vspaceof{ \gv, \hv }^{[i]} =
  \vspaceof{ \gv^{[i]}+\gamma \hv^{[i]} }.
\]
By applying the $i$--th inverse Frobenius to the previous result, we obtain the space $\vspaceof{ \gv + \gamma^{q^{-i}} \hv }$ for any
$i \in \{0, \ldots, r\}$.
In summary, we know the spaces
\[
  \vspaceof{ \gv + \gamma \hv },
  \vspaceof{ \gv + \gamma^{q^{-1}} \hv },
  \ldots,
  \vspaceof{ \gv + \gamma^{q^{-r}} \hv }.
\]

In addition, from Lemma~\ref{lem:supports}, the vector $\av$ and hence
the pair $(\gv, \hv)$ is determined up to some multiplicative
constant. Therefore, one can choose an arbitrary element of
$\vspaceof{ \gv + \gamma \hv }$ and suppose that this element is
$\gv + \gamma \hv$.

\subsection{Step 2: Finding $\gamma$}
In summary, the vector $\gv + \gamma \hv$ and the spaces
$\vspaceof{ \gv + \gamma^{q^{i}} \hv}$ for any
$i \in \{-1, \ldots, -r\}$ are known. To compute $\gamma$, we will use the
following lemma.

\begin{lemma}
  For $i, j \in \{1, \ldots, r\}$, $i \neq j$, there exists a unique
  pair
  $(\uv_{ij},\vv_{ij}) \in \vspaceof{ \gv + \gamma^{q^{-i}} \hv }
  \times \vspaceof{ \gv + \gamma^{q^{-j}} \hv }$
  such that $\uv_{ij} + \vv_{ij} = \gv + \gamma \hv$.
\end{lemma}

\if\longversion1
\begin{proof}
  It suffices to observe that
  $\vspaceof{ \gv, \hv } = \vspaceof{ \gv + \gamma^{q^{-i}} \hv }
  \oplus \vspaceof{ \gv + \gamma^{q^{-j}} \hv }$. \qed
\end{proof}
\fi

The pairs of vectors $(\uv_{ij}, \vv_{ij})$ can be easily computed. Thus, from
now on, we suppose we know them. In addition,
despite $\gamma$, $\gv + \gamma^{q^{-i}}\hv$ and $\gv + \gamma^{q^{-j}}\hv$ 
being unknown, a calculation allows us to show that
$\uv_{ij}, \vv_{ij}$ have the following expressions.
\begin{equation}\label{eq:expr_uij}
  \uv_{ij} = \frac{\gamma^{q^{-j}}-\gamma}{\gamma^{q^{-j}}- \gamma^{q^{-i}}}
  \cdot(\gv + \gamma^{q^{-i}} \hv) \quad {\rm and} \quad
  \vv_{ij} = \frac{\gamma - \gamma^{q^{-i}}}{\gamma^{q^{-j}} - \gamma^{q^{-i}}}
  \cdot (\gv + \gamma^{q^{-j}}\hv).
\end{equation}

Consider the vectors $\uv_{12}$ and $\uv_{13}$. They are collinear
since, from~(\ref{eq:expr_uij}), they are both multiples of
$\gv + \gamma^{q^{-1}}\hv$. Therefore, one can compute the scalar $\alpha$
such that $\uv_{12} = \alpha \cdot \uv_{13}$. From~(\ref{eq:expr_uij}) we deduce
that $\gamma$ satisfies the following relation.
\begin{equation}\label{eq:relation_gamma}
  \frac{\gamma^{q^{-2}} - \gamma}{\gamma^{q^{-2}} - \gamma^{q^{-1}}} =
  \alpha \cdot
  \frac{\gamma^{q^{-3}} - \gamma}{\gamma^{q^{-3}} - \gamma^{q^{-1}}}\cdot
\end{equation}
Or equivalently, $\gamma$ is a root of the polynomial
\[
  Q_\gamma(X) \eqdef  (X^q - X^{q^3})(X - X^{q^2}) - \alpha^{q^3}
  (X - X^{q^3})(X^q - X^{q^2}).
\]
Moreover $(X^q - X)^{q+1}$ divides $Q_\gamma$. Indeed, $X^q - X$
divides $X^{q^2}-X$ and $X^{q^3}-X$ then $(X^q - X)^q$ divides
$X^{q^3}-X^q$ and $X^{q^2}-X^q$.

We now set
\[
  P_\gamma (X) \eqdef \frac{Q_\gamma}{(X^q - X)^{q+1}}\cdot
\]
Since the element $\gamma$ we look for is not in $\Fq$, it is not a
root of $X^q - X$ and hence is a root of $P_\gamma$. Next, the
forthcoming Proposition~\ref{prop:roots_of_P} provides the description
of the other roots. We first need a technical lemma.

\begin{lemma}\label{lem:modular}
  Let $a, b, c, d \in \Fq$, let $i, j$ be two non-negative integers and set
  $A(X) = X^{q^i} - X^{q^j} \in \Fq[X]$. Then,
  \[
    A \left( \frac{aX+b}{cX+d} \right) =
    \frac{ad-bc}{(cX+d)^{q^i+q^j}} \cdot A(X).
  \]
\end{lemma}

\begin{proposition}\label{prop:roots_of_P}
  The set of roots of $P_\gamma$ equals the orbit of $\gamma$ under the action
  of $\PGL{2}{\F_q}$. Equivalently, any root of $P_\gamma$ is of the form
    $\frac{a \gamma + b}{c \gamma + d}$ for $a, b, c, d \in \Fq$ such that
    $ad -bc \neq 0$.
\end{proposition}

\if\longversion1
\begin{proof}
  First, notice that $\deg Q_\gamma = q^3 + q^2$
  and hence
  \[
    \deg P_{\gamma} = \deg Q_\gamma - q(q+1) = q^3 - q = |\PGL{2}{\Fq}|.
  \]
  Second, for any $
  \begin{pmatrix}
    a & b \\ c & d
  \end{pmatrix} \in \PGL{2}{\Fq}$, Lemma~\ref{lem:modular} entails
  \[
    P_\gamma\left( \frac{aX+b}{cX+d}\right) = \frac{1}{(cX+d)^{q^3-q}}
    \cdot P_\gamma(X).
  \]
  Since $\gamma \in \Fqm\setminus\Fq$, then $c\gamma + d \neq 0$ and hence,
  $P_\gamma \left(\frac{a\gamma+b}{c\gamma + d}\right) = 0$.

  We proved that any element in the orbit of $\gamma$ under
  $\PGL{2}{\Fq}$ is a root of $P_\gamma$.  To conclude, we need to
  prove that the orbit of $\gamma$ under $\PGL{2}{\F_q}$ has
  cardinality $\deg P_\gamma = q^3-q$ which means that the stabiliser of
  $\gamma$ with respect to this group action is trivial. Indeed,
  suppose that
  $$
  \gamma = \frac{a \gamma + b}{c\gamma +d},
  \quad \textrm{for\ some}\quad
    \begin{pmatrix}
    a & b \\ c & d
  \end{pmatrix}
\in \PGL{2}{\Fq} \setminus \left\{
  \begin{pmatrix}
    1 & 0 \\ 0 & 1
  \end{pmatrix}
\right\}.
$$
Then $\gamma$ is a root of the polynomial
\[
  X(cX + d) - (aX + b) = cX^2 + (d-a)X + b \in \Fq[X].
\]
This polynomial is nonzero. Indeed, if it was, we would have
$b=c=0$ and $a = d$ which means
\[
  \begin{pmatrix}
    a & b \\ c & d
  \end{pmatrix}
  \sim
  \begin{pmatrix}
    1 & 0 \\ 0 & 1
  \end{pmatrix}
  \quad {\rm in}\quad \PGL{2}{\Fq}.
\]
Next, this nonzero polynomial cancelling $\gamma$ has degree at most
$2$, while from Assumption~(\ref{item:assump2}) the minimal polynomial of
$\gamma$ over $\Fq$ has degree $m > 2$. \qed
\end{proof}
\fi

Thanks to Propositions~\ref{prop:triples} and~\ref{prop:roots_of_P},
we deduce that choosing an arbitrary root $\gamma'$ of $P_\gamma$
provides a candidate for $\gamma$ and there remains to compute
$\gv', \hv'$ providing our triple.
\if\longversion1
Since
$\gamma = \frac{a\gamma' + b}{c \gamma' + d}$ for some
$a, b, c, d \in \Fq$, then, 
\[
\gv + \gamma \hv = \gv' + \gamma' \hv'
\]
where
\[
  \gv' \eqdef \frac{1}{(c \gamma' + d)} \cdot (d\gv + b\hv) \qquad
  {\rm and} \qquad
  \hv' \eqdef \frac{1}{(c \gamma' + d)} \cdot (c\gv + a\hv)
\]
Considering~(\ref{eq:expr_uij}) and using Lemma~\ref{lem:modular}, we get
  \[
    \uv_{12}  =
     \frac{\gamma'^{q^{-2}} - \gamma'}
               {\gamma'^{q^{-2}} - \gamma'^{q^{-1}}}
               \cdot (\gv' + \gamma'^{q^{-1}} \hv').
  \]
Consequently, we know $\gamma'$ and the vectors
$\gv' + \gamma' \hv' = \gv + \gamma \hv$ and $\uv_{12}$.
Thus, we can also compute
\[
  \gv' + \gamma'^{q^{-1}} \hv' =
  \frac{\gamma'^{q^{-2}} - \gamma'^{q^{-1}}}{\gamma'^{-2} - \gamma'} \uv_{12}.
\]
Knowing $\gamma'$, $\gv' + \gamma' \hv'$ and
$\gv' + \gamma'^{q^{-1}} \hv'$ allows us to recover $(\gv', \hv')$.
\else The triple can be deduced from the knowledge of
$\gv + \gamma \hv = \gv' + \gamma' \hv'$ and the computation of
$\gv' + \gamma'^{q^{-1}}\hv'$ which can be proved to satisfy
\[
  \gv' + \gamma'^{q^{-1}} \hv' =
  \frac{\gamma'^{q^{-2}} - \gamma'^{q^{-1}}}{\gamma'^{-2} - \gamma'} \uv_{12}.
\]
\fi 
\subsection{End of the attack}\label{ss:end}
Choose an arbitrary support vector $\av' \in \Fqm^n$ of rank $n$.  Let
$\Qm_0$ (resp. $\Qm_1$) be the unique $n \times n$ matrix with entries in $\Fq$
such that $\av' \Qm_0 = \gv'$ (resp. $\av' \Qm_1 = \hv'$).
Then, set $\Qm \eqdef \Qm_0^T + \gamma' \Qm_1^T$, we have
$$
\Cpubd = \Gab{n-k}{\av'}\cdot \Qm^T
$$
and the matrix $\Qm$ is nonsingular. Indeed, it sends the full rank
vector $\av'$ onto $\gv' + \gamma' \hv' = \gv + \gamma \hv = \av \Pm$.
Therefore, using Lemma~\ref{lem:hpub_decomposition} we get another
representation of the public code as
\[
  \Cpub = \Gab{k}{\av''} \cdot \Qm^{-1},
\]
where $\av'' \in \Fqm^n$ is so that $\Gab{k}{\av''} = \Gab{n-k}{\av'}^\perp$.
Then, any cipher text is of the form $\cv + \ev$ with $\cv \in \Cpub$ and
$\ev \in \Fqm^n$ of rank less than or equal to $\frac{n-k}{4}$. Then
$\ev \Qm$ has rank weight at most $\frac{n-k}{2}$ and hence, the vector
$(\cv + \ev) \Qm$ can be decoded as a corrupted codeword of $\Gab{k}{\av''}$
in order to recover $\ev$ and deduce the plain text.

\subsection{Complexity of the attack}
Let us conclude with a short complexity analysis of the attack.  Let
$\omega$ be the exponent of the complexity of linear algebra
operations. Additions, multiplications will be considered as
elementary operations in $\Fqm$ that we will count.
The evaluation of the Frobenius map costs $O(\log q)$
operations.

\paragraph{\bf Step 1.}  The computation of the dual public code
  $\Cpub^\perp$ costs $O(n^\omega)$ operations.  Next, the computation
  of $\frob{\Cpubd}{1}$ costs $O(n^2 \log q)$ operations in $\Fqm$ and
  the iterative computation of $\frob{\Cpubd}{i}$ for
  $i = 1,\dots, n-k+1$ costs $O(n^3 \log q)$ operations in $\Fqm$.

  The
  computation of $\Cpubd + \frob{\Cpubd}{1}$ boils down to
  Gaussian elimination and hence costs $O(n^\omega)$ operations in
  $\Fqm$. Since we compute $O(n)$ intersections of spaces, the
  overall cost of the computation of
  \[(\Cpubd + \frob{\Cpubd}{1}) \cap (\frob{\Cpubd}{1} + \frob{\Cpubd}{2}) \cap \cdots
  \cap (\frob{\Cpubd}{r} + \frob{\Cpubd}{r+1}),
\]
is of $O(n^{\omega + 1})$ operations in $\Fqm$.  As a conclusion, the
overall cost of the first step is $O(n^3 \log q + n^{\omega +1})$.

\paragraph{\bf Step 2.} The computation of a pair $(\uv_{ij}, \vv_{ij})$
  represents the resolution of a linear system with $2$ unknowns and
  $n$ equations, which costs $O(n)$ operations. This computation
  should be performed $O(n)$ times, yielding an overall cost of
  $O(n^2)$ operations in $\Fqm$, which is negligible compared to the
  previous step.

  Next, the computation of a root of $P_\gamma$ can be computed using
  the Cantor--Zassenhaus algorithm whose complexity is in
  $\widetilde O((\deg P_\gamma)^2 m \log q)$ operations in $\Fqm$
  (see for instance \cite[Th\'eor\`eme 19.20]{BCGLLSS17}),
  where $\widetilde O(\cdot)$ means that the factors in
  $\log (\deg P_\gamma)$ are neglected. Furthermore, since
  $\deg P_\gamma = q^3 - q$ we get a complexity of
  $\widetilde O (m q^6)$ for the calculation of $\gamma'$.

  The remainder of the attack consists in a finite number of linear
  systems solving, i.e. a cost of $O(n^\omega)$, which is negligible
  compared to Step~1.

\paragraph{\bf Summary.}  
This yields an overall cost of
$ O(n^3 \log q + n^{\omega + 1}) + \widetilde O (m q^6)$. Classically,
one chooses $q$ small and $n$ close to $m$, for instance $q=2$ and
$m = O(n)$. In this situation, we get an overall cost of
$O(n^{\omega +1})$.

\subsection{Implementation}
The attack has been implemented using {\sc Magma} and permits to
recover a $4$--tuple $(\av, \gamma' \Qm_0, \Qm_1)$ such that
$\Cpub^\perp = \Gab{n-k}{\av'} \cdot (\Qm_0 + \gamma' \Qm_1)^T$.  The
attack ran on a personal machine\footnote{ Processor:
  Intel\textregistered{} Core\texttrademark{} i5-8250U CPU @ 1.60GHz.
} and succeeded in a few seconds for parameters of cryptographic size
as illustrated by Table~\ref{tab:timings}. Our implementation is only
a proof of concept and could be significantly optimised.
\begin{table}[!h]
  \centering
  \begin{tabular}{|c|c|c|c|c|}
    \hline
    $q$ & $m$ & $n$ & $k$ & {Average time}\\
    \hline
    2 & 50 & 50 & 32 & 0.6 s \\
    \hline
    2 & 80 & 70 & 41 & 1.2 s \\
    \hline
    2 & 120 & 110 & 65 & 9.5 s \\
    \hline
  \end{tabular}
  \medskip
  \caption{Timings for our implementation. For any parameter,
    the attack has been run 100 times and the last column gives the average
  timing.}
  \label{tab:timings}
\end{table}

 \section*{Conclusion}
We provided a distinguisher {\em \`a la Overbeck} for the public keys
of Loidreau's scheme when $\lambda = 2$ and the public code has rate
$R_{\rm pub} \geq \frac 1 2$.  From this distinguisher, we are able to
derive a polynomial time key recovery attack. For small values of $q$,
the attack runs in $O(n^{\omega +1})$ operations in $\Fqm$ and a {\sc
  Magma} implementation succeeds to recover the hidden structure of the
public key in a few seconds.

For larger values of $\lambda$, the distinguisher holds when the dimension of
the code satisfies
\[
n \left(1 - \frac{1}{\lambda}\right)  + 1  < k < n - \lambda.
\]
Here, an analogy with Hamming metric and the original McEliece
encryption scheme can be drawn. Indeed, similarly to Loidreau's public
keys, high rate alternant codes have been proved to be distinguishable
from random codes in polynomial time \cite{FGOPT10}.

\section*{Acknowledgements}
The authors express a deep gratitude to the anonymous referees for
their very relevant suggestions.  The second author was funded by
French ANR projects ANR-15-CE39-0013 {\em Manta} and ANR-17-CE39-0007
{\em CBCrypt}.

\bibliographystyle{splncs04}

\if\longversion1
\appendix
\section{Proof of Proposition~\ref{prop:random}}
\label{subsec:proof_sum_frob}

\subsection{Preliminaries on Gaussian binomial coefficients}

\begin{notation}
  In what follows, we denote by $\GB{a}{b}{q^m}$ the Gaussian binomial
  coefficient representing the number of subspaces of dimension $b$
  of a vector space of dimension $a$ over $\Fqm$.
\end{notation}

\begin{lemma}\label{lem:bounding_GB}
  There exists a positive constant $C$ such that for any
  pair of positive integers $n, k$ such that $n \geq k$,
  we have
    \[
      q^{k(n-k)} \leq \GB{n}{k}{q} \leq C \cdot q^{k(n-k)} .
    \]
\end{lemma}

\begin{proof}
  By definition of Gaussian binomials, we have
  \[
    \GB{n}{k}{q} = \prod_{t = 0}^{k-1}\frac{q^n-q^t}{q^k-q^t} =
    q^{k(n-k)} \prod_{t = 0}^{k-1} \frac{1 - q^{t-n}}{1-q^{t-k}}\cdot
  \]
  Since $n \geq k $, we get
  \[
    \prod_{t=0}^{k-1} \frac{1-q^{t-n}}{1-q^{t-k}} \geq 1,
  \]
  which yields the left hand inequality.  To get the other equality,
  we need to bound from above the product:
  \[
    \prod_{t=0}^{k-1} \frac{1-q^{t-n}}{1-q^{t-k}} \leq \prod_{t=0}^{k-1}
    \frac{1}{1-q^{t-k}} = \prod_{j=0}^{k-1} \frac{1}{1 - \frac{1}{q^{j+1}}},
  \]
  where the last equality is obtained by applying the change of variables
  $j = k-1-t$.
  Set
  \[
    a_k \eqdef \prod_{j=0}^{k-1} \frac{1}{1 - \frac{1}{q^{j+1}}}\cdot
  \]
  The sequence $a_k$ is increasing and converges.
  Indeed,
  \[
    \log (a_k) = \sum_{j=0}^{k-1} - \log \left(1 - \frac{1}{q^{j+1}}\right)
  \]
  and the series with general term $-\log(1 - 1/q^{j+1})$ converges.
  As a conclusion, the right-hand inequality is obtained by taking
  \[
    C \eqdef \prod_{j=0}^{\infty} \frac{1}{1-\frac{1}{q^{j+1}}}\cdot
  \]
  \qed
\end{proof}

\begin{remark}
  A finer analysis would permit to prove that
  $C \leq {\left(\frac{q}{q-1} \right)}^{\frac{q}{q-1}}$. In particular, since
  $q \geq 2$, we have that $C \leq 4$.
\end{remark}

\subsection{The proof}

Let $\Crand$ be a subspace of $\Fqm^n$ chosen uniformly at random among its subspaces of dimension $k$.
From $\Crand$ we build the map
\[
\Psi :\map{\overbrace{\Crand \times \cdots \times \Crand}^{(s+1)\
    \textrm{times}}}{ \F_{q^m}^n}{(\cv_0, \ldots, \cv_s)}{\cv_0 +
  \frob{\cv_1}{1} + \cdots + \frob{\cv_s}{s}}.
\]
The image of this map is
$\Crand + \frob{\Crand}{1} + \cdots + \frob{\Crand}{s}$ and hence the
dimension of $\Crand + \frob{\Crand}{1} + \cdots + \frob{\Crand}{s}$
is related to the dimension of the kernel of $\Psi$. Therefore, our
approach will consist in estimating $\esp{|\ker \Psi|}$.

We have
\begin{equation}\label{eq:esp_Ker_Psi}
  \esp{|\ker \Psi|} = \sum_{\stackrel{(\xv_0, \ldots, \xv_{s}) \in
      {(\F_{q^m}^n)}^s}{\xv_0 + \frob{\xv_1}{1} + \cdots +
      \frob{\xv_s}{s} =0}} \prob\left(\xv_0, \ldots, \xv_{s} \in
    \Crand \right).
\end{equation}

\begin{lemma}
  \label{lem:prob_subset_Crand}
  Let $\AC$ be a subspace of $\Fqm^n$ of dimension $t\leq k$. Then
  \[
  \prob\left(\AC\subseteq\Crand\right) \leq C \cdot q^{-mt(n-k)},
\]
where $C$ is the constant of Lemma~\ref{lem:bounding_GB}.
\end{lemma}

\begin{proof}
  We have
  \[
  \prob (\code A \subseteq \Crand) = \frac{\GB{n- t}{k -
      t}{q^m}}{\GB{n}{k}{q^m}}\cdot
\]
  Using Lemma~\ref{lem:bounding_GB} we get the upper bound,
\[
  \prob (\code A \subseteq \Crand) \leq C\cdot q^{m(k-t)(n-k)} \cdot q^{-mk(n-k)}
  = C \cdot q^{-mt(n-k)}.
\]
  \qed
\end{proof}

For any $0 \leq t \leq s+1$, we introduce the set
\renewcommand\arraystretch{1.3}
\[
  E_t \eqdef \left\{
    (\xv_0, \dots, \xv_s) \in (\Fqm^n)^s ~\Bigg|~
  \begin{array}{rcl}
    \xv_0 + \frob{\xv_1}{1} +\cdots + \frob{\xv_s}{s} &=& 0\\
    \dimfqm \vspaceof{\xv_0, \ldots, \xv_{s}} &=& t    
  \end{array}    
    \right\}\cdot
  \]
  \renewcommand\arraystretch{1}
Thanks to (\ref{eq:esp_Ker_Psi}) and
Lemma~\ref{lem:prob_subset_Crand}, we can write that
\begin{equation}\label{eq:major_esp}
  \esp{|\ker\Psi|}
    \leq
    C \cdot \sum_{t = 0}^{s + 1}
    q^{-mt(n-k)}
     \card{E_t}.
\end{equation}

   \begin{lemma}\label{lem:maj_card_Et}
     Let $1 \leq t \leq k-1$. Then,
     \[
       \card{E_t}\leq C \cdot q^{(mt+n)(s+1-t)+mn(t-1)}.
     \]
   \end{lemma}

   \begin{proof}
     Let $(\xv_0, \ldots, \xv_s) \in E_t$.
     Since the $\xv_i$'s span a
     space of dimension $t$, there exists a unique
     $(s+1-t)\times (s+1)$ full rank matrix $\mat M$ in reduced echelon
     form with entries in $\Fqm$ such that
     \begin{equation}
       \label{eq:MX}
       \mat M \cdot \left(
         \begin{array}{ccc}
          \null \hspace{.3cm} & \xv_0 & \hspace{.3cm} \null\\
           \hline
           & \vdots & \\
           \hline
           & \xv_s &
         \end{array}
       \right)= 0.
     \end{equation}
     Let us count the number of possible $(s+1)$--tuples
     $(\xv_0, \dots, \xv_s)$ satisfying \eqref{eq:MX} and such that
     $\xv_0 + \frob{\xv_1}{1} +\cdots + \frob{\xv_s}{s} = 0$.  For any
     $1 \leq i \leq n$, we have
\begin{eqnarray}
\label{eq:frob}
  x_{0,i} + x_{1,i}^q + \cdots + x_{s,i}^{q^s} &=& 0\\
\label{eq:linear_M}
{\rm and} \qquad  \mat M\cdot
  \left(
\begin{array}{c}
  x_{0,i}\\
  \vdots\\
  x_{s,i}
\end{array}
\right)&=&0
\end{eqnarray}
Let us label the columns of $\mat M$ from $0$ to $s$.  Let
$\mathcal P \subseteq \{0, \ldots, s\}$ be the set of indices of
columns which are pivots for $\mat M$ and $\mathcal{P}^c$ its
complementary. We denote by $a$ the smallest element of
$\mathcal{P}^c$. Notice that
\begin{equation}\label{eq:a0}
  |\mathcal{P}^c|  =t \qquad {\rm and} \qquad a \leq s+1-t.
\end{equation}

In (\ref{eq:frob}), we can eliminate any
$x_{j,i}$ where $j\in \mathcal{P}$ using (\ref{eq:linear_M}).
By this manner we get an expression depending only on the $x_{j,i}$'s
for $j \in \mathcal{P}^c$. If we fix the value of the $x_{r, i}$'s
for any $r \in \mathcal{P}^c \setminus \{a\}$, we obtain an equation of
the form
\[
  Q(x_{a, i}) = 0,
\]
where $Q$ is a $q$--polynomial of $q$--degree at most $a$. There are
then at most $q^{a}$ possible values for $x_{a, i}$ for any choice of
the elements $x_{r,i}$ with $r \in \mathcal{P}^c \setminus \{i\}$.
Using~(\ref{eq:a0}) we deduce that there are at most
\[
  q^{m(t-1) + a} \leq q^{m(t-1) + s+1-t}
\]
possible choices for the tuple $(x_{0, i}, \dots, x_{s,i})$.
Consequently, Equation (\ref{eq:MX}) has at most
$q^{mn(t-1) + n(s+1-t)}$ solutions.

Finally, since the full rank $(s+1-t)\times (s+1)$ matrices in row
echelon form are in one--to--one correspondence with $t$--dimensional
subspaces of $\Fqm^{s+1}$ there are $\GB{s+1}{t}{q^m}$ possible
choices for $\mat M$. Using Lemma~\ref{lem:bounding_GB}, we deduce the
result.  \qed
   \end{proof}

   Combining (\ref{eq:major_esp}) and Lemma~\ref{lem:maj_card_Et}, we get
   \[
     \begin{aligned}
       \esp{|\ker \Psi|} &
       \leq C^2 \cdot \sum_{t=0}^{s+1} q^{-mt(n-k) + (mt+n)(s+1-t)+mn(t-1)}\\
       & \leq C^2 \cdot q^{m(k(s+1)-n)} \cdot \sum_{t=0}^{s+1}
       q^{(s+1-t)(mt+n-mk)}.
       \end{aligned}
     \]
     By assumption (in the statement of Proposition~\ref{prop:random}),
     we have $s < k$. Next, since $n \leq m$, we see that the exponents
     in the above sum are all less than or equal $0$.
     More precisely,
     \[
       \begin{aligned}
         \sum_{t=0}^{s+1} q^{(s+1-t)(mt+n-mk)} & \leq
         2 + \sum_{t=0}^{s-1} q^{(s+1-t)(mt+n-mk)}\\
         & \leq 2 + \sum_{t=0}^{s-1}q^{-m(s+1-t)}\\
         & \leq 2 + \sum_{i=0}^{s-1} q^{-m(i+2)} \leq 2 + \frac{q^{-2m}}{1 - q^{-m}}\cdot
       \end{aligned}
     \]
     Consequently, we have the following result.
     \begin{proposition}\label{prop:esp_final}
     There is a positive constant $C'$ such that
     \[
       \esp{|\ker \Psi|} \leq C' \cdot q^{m(k(s+1)-n)}.
     \]
     \end{proposition}

     To conclude the proof of Proposition~\ref{prop:random}, suppose that
     \[\dim_{\Fqm} \Crand + \frob{\Crand}{1}+\cdots
       +\frob{\Crand}{s} \leq \min \{k(s+1), n\} - a.
     \]
     This means that
     \[
       \dim_{\Fqm} \ker \Psi \geq \max \{0, k(s+1) - n\} + a.
     \]
     Using Markov inequality together with
     Proposition~\ref{prop:esp_final}, we get
     \[
       \begin{aligned}
         \prob \Big( |\ker \Psi| \geq & \ q^{m(\max \{0, k(s+1) - n\} + a)}
         \Big)\\
         &\qquad\leq C' \cdot \frac{q^{m(k(s+1)-n)}}
         { q^{m(\max \{0, k(s+1) - n\} + a)}}\\
         &\qquad\leq C' \cdot q^{-ma}.
       \end{aligned}
     \]
     This concludes the proof.

 \fi
\end{document}